\documentclass[10pt,letterpaper,final]{article}
\usepackage[letterpaper, total={18cm, 24cm}]{geometry}

\usepackage{authblk}
\usepackage[utf8]{inputenc}
\usepackage{amsmath,amsfonts}

\usepackage{stmaryrd}
\usepackage{dsfont}

\usepackage{comment}
\usepackage{xcolor}
\usepackage{array}
\usepackage{amsthm}

\usepackage{bm}
\usepackage{hyperref}
\hypersetup{hidelinks}
\newtheorem{theorem}{Theorem}%
\newtheorem{lemma}[theorem]{Lemma}

\newtheorem{definition}{Definition}
\newtheorem{proposition}{Proposition}

\newtheorem{assumption}{Assumption}
\newtheorem{remark}{Remark}

\definecolor{revXIII}{HTML}{D45500}
\newcommand{\Arxiv}[1]{}
\newcommand{\Rev}[1]{\textcolor{black}{#1}}

\usepackage{mathtools}
\mathtoolsset{showonlyrefs}

\newcommand{\Elio}[1]{{#1}}

\newcommand{\sgn}[2]{\left\lceil#1\right\rfloor^{#2}}
\newcommand{\homm}[2]{\left\llbracket#1\right\rrbracket^{#2}}
\newcommand{\normm}[2]{\left\llbracket#1\right\rrbracket_{#2}}

\newcommand{\mfs}[1]{{\normalfont\textsf{#1}}}
\newcommand{\mf}{\mathbf}

\title{Exact Leader Estimation: A New Approach for Distributed Differentiation\footnote{\textcolor{red}{This is the pre-peer reviewed version of the following article: R. Aldana-López, D. Gómez-Gutiérrez, E. Usai, H. Haimovich, ``Exact Leader Estimation: A New Approach for Distributed Differentiation", Int J Robust Nonlinear Control, 2025, which has been published in final form at \url{https://doi.org/10.1002/rnc.7901}. This article may be used for non-commercial purposes in accordance with Wiley Terms and Conditions for Use of Self-Archived Versions. \textbf{Please cite the publisher's version}. For the publisher's version and full citation details please see: \url{https://doi.org/10.1002/rnc.7901}}}}

\author[1,2]{Rodrigo Aldana-López}
\author[1,3]{David Gómez-Gutiérrez*}
\author[4]{Elio Usai}
\author[5]{Hernan Haimovich}

\date{}

\affil[1]{Intel Tecnología de M\'exico, Intel Labs, Av. del Bosque 1001, 45019, Zapopan, Jalisco, Mexico.}

\affil[2]{Cinvestav, Unidad Guadalajara, Av. del Bosque 1145, 45019, Zapopan, Jalisco, Mexico.}

\affil[3]{Tecnológico Nacional de México, Instituto Tecnológico José Mario Molina Pasquel y Henríquez, Cam. Arenero 1101, 45019, Zapopan, Jalisco, Mexico.}

\affil[4]{University of Cagliari, Department of Electrical and Electronic Engineering, Piazza d'Armi, 09123 Cagliari, Italy.}

\affil[4]{Centro Internacional Franco-Argentino de Ciencias de la Informaci\'on y de Sistemas (CIFASIS) CONICET-UNR, Ocampo \& Esmeralda, 2000, Rosario, Argentina.}

\begin{document}
\maketitle
\begin{abstract}
A novel strategy aimed at cooperatively differentiating a signal among multiple interacting agents is introduced, where none of the agents needs to know which agent is the leader, i.e. the one producing the signal to be differentiated. Every agent communicates only a scalar variable to its neighbors; except for the leader, all agents execute the same algorithm. The proposed strategy can effectively obtain derivatives up to arbitrary $m$-th order in a finite time under the assumption that the $(m+1)$-th derivative is bounded. The strategy borrows some of its structure from the celebrated homogeneous robust exact differentiator by A. Levant, inheriting its exact differentiation capability and robustness to measurement noise. Hence, the proposed strategy can be said to perform robust exact distributed differentiation. In addition, and for the first time in the distributed leader-observer literature, sampled-data communication and bounded measurement noise are considered, and corresponding steady-state worst-case accuracy bounds are derived. The effectiveness of the proposed strategy is verified numerically for second- and fourth-order systems, i.e., for estimating derivatives of up to first and third order, respectively.
\end{abstract}

\footnotetext{\textbf{Abbreviations:} RED (Robust Exact Differentiator).}

\section{Introduction}
\label{Sec:Intro}

A distributed leader-observer is a neighbor-based algorithm that a group of interacting agents implements so that each one obtains an estimate of the leader's state, even when only a subset of them directly communicates with the leader \Elio{and no agent has information on whether it is communicating with the actual leader or not}. \Rev{This problem is important in practice, for instance in formation control of multi-agent systems, since agents often lack direct access to the leader due to obstacles or limited range. Enabling each agent to estimate the leader’s state through neighbor interactions ensures robust coordination, even if the leader changes, without needing centralized control or reconfiguration.} In this case, an unknown but bounded input signal is assumed to drive the leader. Such observers are typically based on distributed consensus algorithms~\cite{cao2012overview,yan2021observer} and are often referred to as distributed consensus observers. Distributed leader-observers have attracted attention in multi-agent tracking and formation control applications~\cite{hong2008}; when the convergence occurs in a finite time, it allows for the separation of the problem into the consensus-based distributed estimation of the leader's state and the local, not consensus-based, reference tracking control toward the leader's state~\cite{zuo2017fixed,miao2018distributed}.

Numerous strategies for designing a distributed leader-observer have been explored in existing research, which are summarized in Table \ref{tab:methods}. While many works focus on arbitrary order leader dynamics~\cite{shi2015,zuo2017fixed,li2020finite,ghommam2021,wang2022,zhang2023,wang2024,wang2022a,chang2022,xu2013,aldana2022dynamic}, some restrict their analysis to second-order systems~\cite{trujillo2020autonomous,ajwad2021,zhao2013,hong2008,zhang2014,cheng2013,hu2010,ni2021,miao2018distributed,li2011finite,zhao2016} due to their relevance in various mechanical applications where double-integrator dynamics are common. These methods differ primarily in the type of information shared between neighbors. In some instances, the entire observer state is shared, as in \cite{ajwad2021,ni2021,miao2018distributed,li2011finite,trujillo2020autonomous} for second-order systems and \cite{cao2015,shi2015,zuo2017fixed,li2020finite,ghommam2021,wang2022,zhang2023,wang2024,wang2022a,Li2018b} for arbitrary order. Alternatively, sharing only the observer's scalar output, which usually represents the leader's estimated position, is another method applied in \cite{zhao2013,hong2008,zhang2014,cheng2013,hu2010} for second-order systems and \cite{zhao2016,chang2022,xu2013,aldana2022dynamic} for arbitrary order. This approach is often preferred since it simplifies the observer structure and reduces the communication load in the protocol.

Additionally, the observer structures vary depending on assumptions about the leader's dynamics. A common assumption is to have autonomous dynamics without external inputs, as seen in \cite{ajwad2021,zhao2013} for second-order systems and \cite{zhao2016,cao2015,chang2022} for arbitrary order, suitable for modeling simple oscillators. However, this assumption can be restrictive for some applications. Considering arbitrary leader inputs under the premise that all followers have access to this information has been explored. However, it compromises the distributed nature of the observer, as seen in \cite{hong2008,zhang2014,cheng2013,hu2010} for second-order systems and \cite{xu2013} for arbitrary order. A more general approach assumes unknown yet bounded leader inputs, enabling the use of sliding mode techniques to design consensus-based observers~\cite{pilloni2023variable} that remain effective under these conditions, as applied in \cite{ni2021,miao2018distributed,li2011finite} for second-order systems and \cite{shi2015,zuo2017fixed,li2020finite,ghommam2021,wang2022,zhang2023,wang2024,wang2022a,aldana2022dynamic} for arbitrary order.

Furthermore, while most observer proposals and analyses are performed in continuous time, the practical implementation of observers typically occurs on a digital computing unit. It involves information sharing over communication networks, making discrete-time sampled-data models more realistic and appropriate for analysis and design. Efforts to develop sampled-data observers have been published, for instance, in \cite{xu2013,Cheng201441}, though they often assume global knowledge of some leader's information. In contrast, recent studies have adapted continuous-time models to incorporate discrete communication using event-triggered methods, as in \cite{wang2022a,chang2022}.

Moreover, it is common for neighbors to access or derive measurements of the leader's position through sensors, potentially introducing noise. Although some studies mention measurement noise or conduct experiments to assess performance under such conditions, like \cite{aldana2022dynamic}, formal analysis in this area is generally lacking. An exception is \cite{hu2010}, which includes a formal analysis considering Gaussian measurement noise. However, as mentioned before, this work is limited to second-order systems and global knowledge of leader's information.

\Elio{Finally, note that a discretized standard differentiator \cite{livne2014} can be used to estimate the leader's state and its derivatives by agents with access to the leader. These estimations are robust to sampling time and measurement noise. This information can then be passed from agent to agent across the network using a flooding strategy. However, this approach has two significant drawbacks. First, it assumes that agents know which other agent is the leader and applies the differentiator accordingly. Second, multi-hop flooding strategies are not preferred in large multi-agent systems because they consume much more communication bandwidth than consensus protocols, which rely only on local communication.}

Motivated by the above discussion on the current state of the art, we assume a leader with an arbitrary order integrator dynamics driven by an unknown input signal bounded by a known constant, with an output signal as the first variable in the integrator chain, accessed only by a subset of agents \Elio{that are unaware who the leader is, thus rendering a flooding strategy unfeasible}. In this context, the contributions of this work are summarized as follows:
\begin{itemize}
    \item Since the aim is distributed differentiation, no agent requires knowledge of the actual leader's input, i.e. of the highest-order derivative of the function to be differentiated. The only piece of knowledge required from this high-order derivative is its bound.
    \item Every agent communicates only a scalar variable with its neighbors.
    \item For the noise-free and continuous output case, we prove exact convergence to the function to be differentiated and its derivatives (i.e. the leader's state).
    \item Under {measurement noise} affecting each output and {sampled outputs}, the accuracy of the algorithm in terms of differentiation error is analyzed.
    \item The order of differentiation can be arbitrary.
\end{itemize}

Table~\ref{tab:methods} summarizes and highlights the contribution of the paper with respect to the cited literature.

\begin{table}
\centering
\normalfont{
\begin{tabular}{m{3cm} m{1.5cm} m{2.5cm} m{2.5cm} m{4cm} m{2cm} }
\textbf{Observer} & \textbf{Order} & \textbf{Communication} & \textbf{Time domain} & \textbf{Leader's input type} & \textbf{Meas. noise} \\ 
\hline
\cite{trujillo2020autonomous}           & second & state  & continuous & bounded unknown & no \\ 
\cite{ajwad2021}           & second & state  & continuous & no input        & no \\ 
\cite{zhao2013}            & second & scalar output & continuous & no input        & no \\
\cite{hong2008,zhang2014,cheng2013}            & second & scalar output & continuous & globally known          & no\\
\cite{hu2010}              & second & scalar output & continuous & globally known          & Gaussian\\
\cite{ni2021,miao2018distributed,li2011finite}              & second & state  & continuous & bounded unknown & no \\

\cite{zhao2016}       & arbitrary & scalar output    & continuous         & no input        & no \\
\cite{cao2015,Li2018b}        & arbitrary & state     & continuous         & no input        & no \\ 

\cite{shi2015,zuo2017fixed,li2020finite,ghommam2021,wang2022,zhang2023,wang2024}   & arbitrary & state     & continuous         & bounded unknown & no \\

\cite{wang2022a}      & arbitrary & state     & event-triggered    & bounded unknown & no \\
\cite{chang2022}      & arbitrary & scalar output    & event-triggered    & no input        & no \\
\cite{xu2013}         & arbitrary & scalar output    & sampled-data       & globally known          & no \\
\cite{aldana2022dynamic}& arbitrary& scalar output plus aux. variable &continuous & bounded unknown & no \\
This work& arbitrary& scalar output& cont./sampled-data& bounded unknown& bounded noise
\end{tabular}
}
\caption{Summary of distributed leader observers in the literature. }
\label{tab:methods}
\end{table}

The rest of the paper is organized as follows. Section~\ref{Sec:Main} introduces the distributed observer problem and the proposed algorithm to address it. In Section~\ref{Sec:Proofs}, we present the proof of the main Theorem. Also, we analyze the performance under noisy and discrete measurements of the neighbors' and leader's output. In Section~\ref{Sec:Num}, we illustrate the effectiveness of the proposed algorithm via numerical simulations. Finally, in Section~\ref{Sec:Concl}, we present the conclusions and discuss future work. 

\textbf{Notation: } We write vectors in boldface lower case letters and matrices in boldface upper case letters. Let $\sgn{x}{\alpha}:=|x|^\alpha\text{sign}(x)$ for $\alpha\in\mathbb{R}\setminus\{0\}$ and $\sgn{x}{0}:=\text{sign}(x)$. For an $m$-times differentiable signal $y(t)$, $y^{(i)}(t)$, $i\leq m$, denotes its $i$-th time derivative. For a matrix $\mf{A}$, \Elio{$[\mf{A}]_{ij}$ denotes its $i,j$ element}, $\mf{A}^\top$ denotes the matrix transpose. Given a square matrix $\mf{A}$, $\lambda_{\max}(\mf{A})$ denotes the largest singular value of $\mf{A}$, and $\mf{A}\succ \mf{0}$ denotes that $\mf{A}$ is a positive definite matrix.
Moreover, if $f:\mathbb{R}\to\mathbb{R}$ and $\mf{v}=[v_1,\dots,v_\mfs{N}]^\mfs{T}\in\mathbb{R}^\mfs{N}$, then we write $f(\mf{v}):=[f(v_1),\dots,f(v_\mfs{N})]^\top$, in particular, such notation is used for $\sgn{\mf{v}}{\alpha}\in\mathbb{R}^{\mfs{N}}$ and element-wise absolute value $|\mf{v}|\in\mathbb{R}^{\mfs{N}}$. Furthermore, for two vectors $\mf{v},\mf{w}\in\mathbb{R}^\mfs{N}$, $\mf{v}\odot\mf{w}:=\text{diag}(\mf{v})\mf{w}\in\mathbb{R}^\mfs{N}$ denotes the element-wise multiplication.
Denote with \Elio{$\|\mf{v}\|:=\sqrt{\mf{v}^\top \mf{v}}$} and
$
    \llbracket\mf{v}\rrbracket^\alpha := \sum_{i=1}^\mfs{N}|v_i|^\alpha
$
for arbitrary vector $\mf{v}=[v_1,\dots,v_\mfs{N}]^\top$ and $\alpha\geq 0$, as well as $\normm{\mf{v}}{\alpha}:=(\homm{\mf{v}}{\alpha})^{1/\alpha}$. Moreover, $\mathds{1}$ is a vector of ones of appropriate dimension. For a natural scalar $v$, $v!$ denotes its factorial.

\Elio{We consider a team of $\mfs{N}$ agents, connected in a communication network modeled by an undirected graph $\mathcal{G}=(\mathcal{I},\mathcal{E})$ with vertex set $\mathcal{I}=\{1,\dots,\mfs{N}\}$ and edge set $\mathcal{E}\subseteq\mathcal{I}\times \mathcal{I}$ with $(i,i)\notin\mathcal{E}$. For an agent $i\in\mathcal{I}$, we say that an agent $j\in\mathcal{I}$ is its neighbor, if $(i,j)\in\mathcal{E}$, and we denote its neighbor set by $\mathcal{N}_i=\{j:(i,j)\in \mathcal{E}\}$. 
For a graph $\mathcal{G}$, $\mf{A}$ denotes its adjacency matrix with components $[\mf{A}]_{ij}=1$ if $(i,j)\in\mathcal{E}$ and $[\mf{A}]_{ij}=0$ otherwise, and 
$\mf{L}=\text{diag}(\mf{A}\mathds{1}) - \mf{A}$ denotes its Laplacian matrix.}

\section{Main result}
\label{Sec:Main}
\subsection{Problem Statement: Distributed leader-observer proposal}
Denote with $u(t)\in\mathbb{R}$ the leader's output signal satisfying
$
|u^{(m+1)}(t)|\leq L
$
with known $L\geq 0$ and $m\in\mathbb{N}$. We consider a team of $\mfs{N}$ agents (followers), connected in a communication network modeled by an \textit{undirected graph} $\mathcal{G}$. Each agent $i\in\mathcal{I}$ stores its $m+1$ state variables ${x}_{i,0}(t),\dots,{x}_{i,m}(t)\in\mathbb{R}$, \Elio{and has access to the continuous-time signal ${x}_{j,0}(t)$ of each of its neighbors $j\in\mathcal{N}_i$}. \Elio{Since the communication graph is undirected, agent $j$ also has access to the continuous-time signal ${x}_{i,0}(t)$ from the $i$-th agent}. Moreover, a subset of agents $\mathcal{L}\subseteq\mathcal{I}$ additionally has access to the continuous-time output $u(t)$ of the leader, \Elio{but \textit{the agents have no information on who the leader is}}.

The aim is to design a dynamic algorithm such that, for each $i\in\mathcal{I}$ and every $\mu=0,\dots,m$, $x_{i,\mu}(t)$ converges exactly to the value of the leader's $\mu$-th derivative $u^{(\mu)}(t)$, effectively performing what is called \emph{exact distributed differentiation} in this work. 

To this aim, using the available information, the state variables of each agent are updated according to:
\begin{equation}
\label{eq:main}
\begin{array}{rl}
\dot{x}_{i,\mu} &= {x}_{i,\mu+1}-k_\mu \tilde{L}^{\frac{\mu+1}{m+1}}\sgn{\displaystyle\sum_{j\in\mathcal{N}_i}({x}_{i,0}-{x}_{j,0})+b_i({x}_{i,0}-u)}{\frac{m-\mu}{m+1}}\ \ \text{for\ }  \mu=0,\dots,m-1 \\
\dot{x}_{i,m} &= -k_m\tilde{L}\sgn{\displaystyle\sum_{j\in\mathcal{N}_i}({x}_{i,0}-{x}_{j,0})+b_i({x}_{i,0}-u)}{0} \\
\end{array}
\end{equation}
where $b_i=1$ if $i\in\mathcal{L}$, i.e. if agent $i$ has access to the leader's output, and $b_i=0$, otherwise. Moreover, $\tilde{L}=\mfs{N}L\lambda_{\max}(\mf{H}^{-1}\mf{B})$ where $\mf{H}=\mf{L}+\mf{B}$, where $\mf{B}=\text{diag}(b_1,\dots,b_\mfs{N})$. As established in Lemma \ref{le:H} of Appendix \ref{sec:aux}, $\mf{H}$ is ensured to have an inverse under the following assumption:

\begin{assumption}
\label{as:graph}
The graph $\mathcal{G}$ is undirected, fixed, and connected. Moreover, at least one follower agent has access to the leader.
\end{assumption}

A graph is undirected if $(i,j)\in\mathcal{E}$ implies that $(j,i)\in\mathcal{E}$; fixed if the vertex set $\mathcal{I}$ and edge set $\mathcal{E}$ do not change during the evolution of the experiment; and connected if for every pair of distinct vertices, $i,j\in\mathcal{I}$, there is a sequence of vertices from $i$ to $j$ such that consecutive vertices are neighbors. \Arxiv{The undirected graph assumption is reasonable for an observer of the form \eqref{eq:main} since the peer-to-peer interactions occur through the error terms $x_{i,0}-x_{j,0}$ which arise from a virtual system rather than from physical sensor information. Henceforth, a communication system is required for sharing the output $x_{i,0}$ between neighbors. In this setting, most modern communication protocols enable bi-directional communication, making it feasible to implement \eqref{eq:main}.}

Given that the proposed algorithm has discontinuous right-hand sides, in this paper the solutions are understood in Filippov's sense \cite[Page 85]{Filippov1988DifferentialSides}.

\subsection{Main Theorem}
To establish the main result, we define an appropriate error variable, relating it to the estimation goal. Denote with $\mf{x}_\mu(t)=[{x}_{1,\mu}(t),\dots,{x}_{\mfs{N},\mu}(t)]^\top$ and write \eqref{eq:main} in vector form as:
\begin{equation}
\label{eq:main:vector}
\begin{array}{rl}
\dot{\mf{x}}_{\mu} &= \mf{x}_{\mu+1}-k_\mu \tilde{L}^{\frac{\mu+1}{m+1}}\sgn{\mf{L}\mf{x}_0+\mf{B}(\mf{x}_{0}-\mathds{1}u)}{\frac{m-\mu}{m+1}}\quad \quad  \text{for\ }  \mu=0,\dots,m-1 \\[0.5em]
\dot{\mf{x}}_{m} &= -k_m\tilde{L}\sgn{\mf{L}\mf{x}_0+\mf{B}(\mf{x}_{0}-\mathds{1}u)}{0}.
\end{array}
\end{equation}

Define the error variables 
$$
\mf{e}_\mu(t) = \mf{x}_\mu(t)-\mf{H}^{-1}\mf{B}\mathds{1}u^{(\mu)}(t) \quad \quad \text{for\ }  \mu=0,\dots,m-1.
$$
\begin{proposition}
\label{prop:equilibrium}
    Let Assumption \ref{as:graph} hold. Assume that $\mf{e}_\mu(t)=\mf{0}$ for all $\mu\in\{0,\dots,m\}$ at a given $t\in\mathbb{R}$. Then, it follows that $x_{i,\mu}(t)=u^{(\mu)}(t)$.
\end{proposition}
\begin{proof}
    First, we show that $\mf{H}^{-1}\mf{B}\mathds{1}=\mathds{1}$. To do so, we solve for solutions $\mf{v}\in\mathbb{R}^\mfs{N}$ to the equation $\mf{H}^{-1}\mf{B}\mf{v}=\mf{v}$ which imply
    $\mf{0}=(\mf{B}-\mf{H})\mf{v}=(\mf{B}-\mf{B}-\Elio{\mf{L}})\mf{v}=-\Elio{\mf{L}}\mf{v}$. Hence, $\mf{v}$ is an eigenvector of $\Elio{\mf{L}}$ \Elio{with eigenvalue $0$}, which must be a multiple of $\mathds{1}$ since $\mathcal{G}$ is connected \cite[Page 279]{godsil2001}. Now, $\mf{e}_\mu(t)=\mf{0}$ imply $\mf{x}_\mu=\mf{H}^{-1}\mf{B}\mathds{1}u^{(\mu)}(t)=\mathds{1}u^{(\mu)}(t)$, completing the proof.
\end{proof}
Under the light of Proposition \ref{prop:equilibrium}, we aim to show that the equality $\mf{e}_\mu(t)=0$ is reached and maintained in finite time, enabling \eqref{eq:main} to perform exact distributed differentiation. To establish such result, the dynamics of $\mf{e}_\mu(t)$ are written as:
\begin{equation}
\label{eq:main:e}
\begin{array}{rl}
\dot{\mf{e}}_{\mu} &= \mf{e}_{\mu+1}-k_\mu \tilde{L}^{\frac{\mu+1}{m+1}}\sgn{\mf{H}\mf{e}_0}{\frac{m-\mu}{m+1}}\quad\quad\quad\quad\text{for\ }  \mu=0,\dots,m-1 \\[0.5em]
\dot{\mf{e}}_{m} &=-\mf{H}^{-1}\mf{B}\mathds{1}u^{(m+1)}(t) -k_m\tilde{L}\sgn{\mf{H}\mf{e}_0}{0},
\end{array}
\end{equation}
where we used the fact that
$
\mf{L}\mf{x}_0+\mf{B}(\mf{x}_{0}-\mathds{1}u) = \mf{H}(\mf{x}_0 - \mf{H}^{-1}\mf{B}\mathds{1}u) = \mf{H}\mf{e}_0.
$
Moreover, it follows that $-\mf{H}^{-1}\mf{B}\mathds{1}u^{(m+1)}(t)\in[-\tilde{L},\tilde{L}]^\mfs{N}$ given \Elio{$\mathds{1}u^{(m+1)}(t)\in[-{L},{L}]^\mfs{N}$}. 

\begin{definition} 
Let $\mf{e}=[\mf{e}_0^\top\ \ \cdots \ \ \mf{e}_m^\top]^\top$ and let $\mf{e}(t;\mf{e}(0))$ be the solution of \eqref{eq:main:e} with initial condition $\mf{e}(0)$. Then, the origin of system \eqref{eq:main:e} is said to be \textit{Lyapunov stable} if for all $\epsilon>0$, there is $\delta:=\delta(\epsilon)>0$ such that for all $\|\mf{e}(0)\|<\delta$, any solution $\mf{e}(t;\mf{e}(0))$ of~\eqref{eq:main:e} exists for all $t\geq0$, and $\|\mf{e}(t;\mf{e}(0))\|<\epsilon$ for all $t\geq0$. Moreover, the origin of \eqref{eq:main:e} is said to be \textit{globally finite-time stable} if it is Lyapunov stable, and for any $\mf{e}(0)\in\mathbb{R}^{\mfs{N}(m+1)}$, there exists $0\leq T < +\infty$ such that $\mf{e}(t;\mf{e}(0))=0$ holds for all $t\geq T$. 
\end{definition}

Our main result is the following:
\begin{theorem}
\label{th:main}
    Let Assumption \ref{as:graph} hold. Then, there exist positive parameters $\{k_\mu\}_{\mu=0}^m$ (sufficiently large) such that the origin of system~\eqref{eq:main:e} is finite-time stable. 
    
    In particular:
    
    \begin{itemize}
        \item  if $m=1$ set:
    \begin{equation}
    \label{eq:gains}
    \begin{aligned}
    & k_0>\sup_{\|\mf{z}\|=1} \left(\frac{\gamma_0(\mf{z})}{\eta_0(\mf{z})}\right),\quad k_1>1,
    \end{aligned}
    \end{equation}
    where $\eta_0(\mf{z}), \gamma_0(\mf{z})$ are defined subsequently in \eqref{eq:eta:0}.

    \item If $m> 1$ then the gains can be chosen recursively as $k_\mu=\tilde{k}_\mu k_{\mu-1}^{\frac{m-\mu}{m-(\mu-1)}}$ for $\mu=1,\dots,m$ and $\{\tilde{k}_\mu\}_{\mu=1}^m$  taken from \eqref{eq:levant_recursive} for convergence of a RED \cite{levant2003}, and with sufficiently large $\tilde{k}_0=k_0>0$.
    \end{itemize}
\end{theorem}

\subsection{\Elio{Proof of the main result}}
\label{Sec:Proofs}
\subsubsection{Proof of Theorem~\ref{th:main} for $m=1$}

For this part, we follow a similar Lyapunov-based proof strategy as in \cite{cruz2019} for the RED. Introduce the normalized error variables
$$
\mf{z}_0=\frac{\mf{e}_0}{\tilde{L}},\quad\mf{z}_1=\frac{\mf{e}_1}{k_{0}\tilde{L}}
$$
with error dynamics:
\begin{equation}
\label{eq:main:m:1}
\begin{array}{rl}
\dot{\mf{z}}_{0} &= -\tilde{k}_0\left(\sgn{\mf{H}\mf{z}_0}{\frac{1}{2}}-\mf{z}_{1}\right)\\
\dot{\mf{z}}_{1} &\in-\tilde{k}_1\left(\sgn{\mf{H}\mf{z}_0}{0}+\frac{1}{k_1}[-1,1]^\mfs{N}\right) \\
\end{array}
\end{equation}
where
$
\tilde{k}_0=k_0, \tilde{k}_1=k_1/k_0
$.
Note that \eqref{eq:main:m:1} is an homogeneous differential inclusion for the state $\mf{z}=[\mf{z}_0^\top,\mf{z}_1^\top]^\top$ with  homogeneity degree $-1$ and weights $\mf{r}=[r_0\mathds{1}^\top,r_1\mathds{1}^\top]$ where $r_0=2, r_1=1$ (See Appendix \ref{sec:homo}).

The proof idea is based on the following Lyapunov function:
\begin{equation}
    \label{eq:lyap:m:1}
    V(\mf{z}) = \frac{1}{2}\mf{z}_0^\top\mf{H}\mf{z}_0 - \mf{z}_0^\top\sgn{\mf{z}_1}{2}+\frac{h}{4}\homm{\mf{z}_1}{4}
\end{equation}
for some $h>0$ to be selected later.
\Rev{First, we show that $V(\mf{z})$ is positive definite. This can be verified since \eqref{eq:lyap:m:1} can be written as $V(\mf{z})=\bm{\psi}(\mf{z})^\top\mf{M}(h)\bm{\psi}(\mf{z})$ with:
$$
\bm{\psi}(\mf{z}) = \begin{bmatrix}
\mf{z}_0 \\ \sgn{\mf{z}_1}{2}
\end{bmatrix}, \ \ \mf{M}(h)=
\begin{bmatrix}
\mf{H}/2 & -\mf{I}_\mfs{N}/2 \\
-\mf{I}_\mfs{N}/2 & h\mf{I}_\mfs{N}/4
\end{bmatrix}
$$
Note that the positive definiteness of $V(\mf{z})$ is uniquely defined by the positive definiteness of $\mf{M}(h)$. Lemma \ref{le:schur} in Appendix \ref{sec:aux} ensures that $\mf{M}(h)\succ \mf{0}$ if and only if $\mf{H}\succ \mf{0}$ (ensured by Lemma \ref{le:H} in Appendix \ref{sec:aux}) and  
$
h\mf{I}_\mfs{N}/4 - (\mf{H}/2)^{-1}/4\succ 0
$
which is satisfied if $h>2\lambda_{\max}(\mf{H}^{-1})$.}

In addition, note that $V(\mf{z})$ is $\mf{r}$-homogeneous of degree $4$. Now, analyze the evolution of $V$ along the trajectories of \eqref{eq:main:m:1}:
$$
\begin{aligned}
\dot{V} &= \left(\mf{H}\mf{z}_0-\sgn{\mf{z}_1}{2}\right)^\top\dot{\mf{z}}_0 + \left(2\mf{z}_0\odot |\mf{z}_1| + h\sgn{\mf{z}_1}{3}\right)^\top\dot{\mf{z}}_1 \\
&\in -\tilde{{k}}_0\left(\mf{H}\mf{z}_0-\sgn{\mf{z}_1}{2}\right)^\top\left(\sgn{\mf{H}\mf{z}_0}{\frac{1}{2}}-\mf{z}_{1}\right)  -\tilde{k}_1\left(2\mf{z}_0\odot |\mf{z}_1| + h\sgn{\mf{z}_1}{3}\right)^\top\left(\sgn{\mf{H}\mf{z}_0}{0}+\frac{1}{k_1}[-1,1]^\mfs{N}\right) \\
\dot{V}&\leq-\tilde{k}_0\eta_0(\mf{z}) + \gamma_0(\mf{z})
\end{aligned}
$$
with
\begin{equation}
\label{eq:eta:0}
\begin{aligned}
&\eta_0(\mf{z}) = \left(\mf{H}\mf{z}_0-\sgn{\mf{z}_1}{2}\right)^\top\left(\sgn{\mf{H}\mf{z}_0}{\frac{1}{2}}-\mf{z}_{1}\right)\\
    &\gamma_0(\mf{z})=\Elio{\sup_{\bm{\xi}\in[-1,1]^{\mfs{N}}}-}\tilde{k}_1\left(2\mf{z}_0\odot |\mf{z}_1| + h\sgn{\mf{z}_1}{3}\right)^\top\left(\sgn{\mf{H}\mf{z}_0}{0}+\frac{\bm{\xi}}{k_1}\right) 
\end{aligned}
\end{equation}
Note that $\eta_0(\mf{z})\geq 0$ by Lemma \ref{le:ineq} in Appendix \ref{sec:aux} using $\mf{v} = \mf{H}\mf{z}_0, \mf{w}=\sgn{\mf{z}}{2}, \alpha=1/2$. Henceforth, since $\tilde{k}_1$ is picked arbitrarily independent of $\tilde{k}_0$ as stated in the theorem, existence of $\tilde{k}_0$ such that $-\tilde{k}_0\eta_0(\mf{z})+\gamma_0(\mf{z}_0)<0$ is ensured by Lemma \ref{le:gain} in Appendix~\ref{sec:aux} if $\gamma_0(\mf{z})<0$ when $\eta_0(\mf{z})=0$, i.e., when $\mf{Hz}_0=\sgn{\mf{z}_1}{2}$. Denote with 
$
\mathcal{Z}_0 = \{\mf{z}\in\mathbb{R}^{2\mfs{N}}: \mf{Hz}_0=\sgn{\mf{z}_1}{2}\}
$.
Hence, denoting $\gamma_{0}(\mf{z}|\mathcal{Z}_0)$ the value of $\gamma_0(\mf{z})$ restricted to $\mf{z}\in\mathcal{Z}_0$ leads to:
$$
\begin{aligned}
&\gamma_0(\mf{z}|\mathcal{Z}_0)=\Elio{\sup_{\bm{\xi}\in[-1,1]^{\mfs{N}}}-}\tilde{k}_1\left(2(\mf{H}^{-1}\sgn{\mf{z}_1}{2})\odot |\mf{z}_1| + h\sgn{\mf{z}_1}{3}\right)^\top\left(\sgn{\mf{z}_1}{0}+\frac{\bm{\xi}}{k_1}\right) 
\end{aligned}
$$
Note that:
\begin{equation}
\label{eq:last:state}
\begin{aligned}
&-\left(\sgn{\mf{z}_1}{\alpha}\right)^\top\left(\sgn{\mf{z}_1}{0}+\frac{\bm{\xi}}{k_1}\right) = -\homm{\mf{z}_1}{\alpha} -\frac{1}{k_1}\sum_{i=1}^\mfs{N}\sgn{z_{i,1}}{\alpha}\xi_i \\
&\leq -\homm{\mf{z}_1}{\alpha} + \frac{1}{k_1}\sum_{i=1}^\mfs{N}|z_{i,1}|^\alpha|\xi_i| \leq -\homm{\mf{z}_1}{\alpha}+ \frac{1}{k_1}\sum_{i=1}^\mfs{N}|z_{i,1}|^\alpha  = -\homm{\mf{z}_1}{\alpha}\left(1-\frac{1}{k_1}\right)
\end{aligned}
\end{equation}
for $\alpha=3$ and $\bm{\xi}\in[-1,1]^\mfs{N}$.
Henceforth, denote with:
\begin{equation}
\label{eq:eta:1}
\begin{aligned}
    &\eta_1(\mf{z}_1) = \homm{\mf{z}_1}{\alpha}\left(1-\frac{1}{k_1}\right) \\
    &\gamma_1(\mf{z}_1) = \Elio{\sup_{\bm{\xi}\in[-1,1]^{\mfs{N}}}-}\left(2(\mf{H}^{-1}\sgn{\mf{z}_1}{2})\odot |\mf{z}_1|\right)^\top\left(\sgn{\mf{z}_1}{0}+\frac{\bm{\xi}}{k_1}\right) 
\end{aligned}
\end{equation}
such that:
$$
\gamma_0(\mf{z}|\mathcal{Z}_0) \leq \tilde{k}_1\left(-h\eta_1(\mf{z}_1) + \gamma_1(\mf{z}_1)\right)
$$
Recall that $k_1>1$ such that $\eta_1(\mf{z}_1)>0$ except at $\mf{z}_1=\mf{0}$. Henceforth, applying Lemma \ref{le:gain} in Appendix \ref{sec:aux} again, it is concluded that there exist $h^*\in\mathbb{R}$ such that $-h\eta_1(\mf{z}_1)+\gamma_1(\mf{z}_1)<0$ for all $h>h^*$. Therefore, picking \Elio{$h>\max\{2\lambda_{\max}(\mf{H}^{-1}),h^*\}$} with

    $$
    h^* = \sup_{\|\mf{z}_1\|=1}\left(\frac{\gamma_1(\mf{z}_1)}{\eta_1(\mf{z}_1)}\right)
    $$
    where $\eta_1(\mf{z}_1), \gamma_1(\mf{z}_1)$
ensures both that $V$ is positive definite, radially unbounded and $\dot{V}<0$, establishing global asymptotic stability. This property in addition to homogeneity of \eqref{eq:main:m:1} allows to conclude finite time stability.

\subsubsection{Proof of Theorem~\ref{th:main} for arbitrary $m>1$}
\label{sec:arbitrary}
For this part, we follow a similar homogeneity-based proof strategy as in \cite{levant2003} for the RED. First, introduce the normalized error variables
$$
\mf{z}_\mu=\frac{\mf{e}_\mu}{\tilde{L}}, \mu=0,\dots,m
$$
with error dynamics:
\begin{equation}
\label{eq:main:z}
\begin{array}{rl}
\dot{\mf{z}}_{\mu} &= -{k}_\mu\sgn{\mf{H}\mf{z}_0}{\frac{m-\mu}{m+1}}+\mf{z}_{\mu+1} \quad\quad\quad\quad\text{for\ }  \mu=0,\dots,m-1 \\[0.5em]
\dot{\mf{z}}_{m} &\in-{k}_m\sgn{\mf{H}\mf{z}_0}{0} +[-1,1]^{\mfs{N}}.
\end{array}
\end{equation}
Note that \eqref{eq:main:z} is an homogeneous differential inclusion for the state $\mf{z}=[\mf{z}_0^\top,\dots,\mf{z}_m^\top]^\top$ with  homogeneity degree $-1$ and weights $\mf{r}=[r_0\mathds{1}^\top,\dots,r_m\mathds{1}^\top]$ where $r_\mu=m+1-\mu$ for $\mu=0,\dots,m$. Now, write \eqref{eq:main:z} in recursive form by solving for $\mf{H}\mf{z}_0$ in the $(\mu-1)$-th equation
$$
\mf{H}\mf{z}_0 = k_{\mu-1}^{-\frac{m+1}{m-(\mu-1)}}\sgn{\mf{z}_\mu-\dot{\mf{z}}_{\mu-1}}{\frac{m+1}{m-(\mu-1)}}
$$
and replacing it in the $\mu$-th equation with $\mu\geq 1$:
\begin{equation}
\label{eq:recursive}
\begin{array}{rl}
\dot{\mf{z}}_{0} &= -\tilde{k}_0\sgn{\mf{H}\mf{z}_0}{\frac{m}{m+1}}+\mf{z}_{1}\\[0.5em]

\dot{\mf{z}}_{\mu} &= -\tilde{k}_\mu\sgn{\mf{z}_\mu-\dot{\mf{z}}_{\mu-1}}{\frac{m-\mu}{m-(\mu-1)}}+\mf{z}_{\mu+1}\quad\quad\quad \text{for\ }  \mu=1,\dots,m-1 \\[0.5em]

\dot{\mf{z}}_{m} &\in-\tilde{k}_m\sgn{\mf{z}_m-\dot{\mf{z}}_{m-1}}{0}+[-1,1]^{\mfs{N}} \\
\end{array}
\end{equation}
where $\tilde{k}_0=k_0,\tilde{k}_\mu := k_\mu k_{\mu-1}^{-\frac{m-\mu}{m-(\mu-1)}}$ for $\mu=1,\dots,m$.
Note that \eqref{eq:recursive} corresponds to the recursive form of $\mfs{N}$ homogeneous differentiators, one per component of $\mf{z}_\mu$. To see this, compare with \eqref{eq:levant_recursive} in Appendix \ref{sec:homodif}, where the components of $\dot{\mf{\mf{z}}}_0(t)$ play the role of $\theta(t)$ in each differentiator. Now, the following technical lemmas formally establish this idea by studying the behavior of ${\mf{z}}_0(t)$ in order to conclude a contraction property for \eqref{eq:recursive}. 
\begin{lemma}
\label{le:zerobound}
Let Assumption \ref{as:graph} hold and consider trajectories of \eqref{eq:recursive}. 
Given any $\delta_{\mfs{hom}}>0,B_{\mf{z}_0}^{t_0},\dots,B_{\mf{z}_m}^{t_0}>0$ satisfying $\|\mf{z}_\mu(t_0)\|\leq B_{\mf{z}_\mu}^{t_0}$ then, there exist $B_{\mf{z}_\mu}>B_{\mf{z}_\mu}^{t_0}$ such that  $\|\mf{z}_\mu(t)\|\leq B_{\mf{z}_\mu}, \forall \Elio{t\in[t_0,t_0+\delta_{\mfs{hom}})}$.
\end{lemma}
\begin{proof}
   Existence of the bounds follow due to the absence of finite-time escapes in \eqref{eq:recursive} as a consequence of homogeneity.
\end{proof}

\begin{lemma}
\label{le:state}
Let Assumption \ref{as:graph} hold. Consider the trajectories of
\begin{equation}
    \dot{\mf{z}}_0 = \mf{d}(t)-\tilde{k}_0\sgn{\mf{H}\mf{z}_0}{\frac{m}{m+1}},
\end{equation}
with bounded disturbance $\|\mf{d}(t)\|\leq B_\mf{d}>0, \forall t\geq t_0$ given any $\delta_{\mf{d}},B_{\mf{H}\mf{z}_0}^\infty>0$ then, there exist $\tilde{k}_{0,\mf{d}}\in\mathbb{R}$ such that if $\tilde{k}_0\geq \tilde{k}_{0,\mf{d}}$, then $\|\mf{Hz}_0(t)\|\leq B_{\mf{H}\mf{z}_0}^\infty$ for all $t\geq t_0+\delta_\mf{d}$. Consequently, given $\tilde{k}_0, \delta_\mf{d}>0$, there exist a constant $c_0>0$  depending on $\delta_{\mf{d}}$ such that $\|\mf{Hz}_0(t)\|\leq c_0 \tilde{k}_0^{-\frac{m+1}{m}}$ for all $t\geq t_0+\delta_\mf{d}$.
\end{lemma}
\begin{proof}
    Choose a Lyapunov function $V(\mf{z}_0) = \frac{1}{2}\mf{z}_0^\top\mf{H}\mf{z}_0$. Therefore, for any $t\geq t_0$ such that:
    $$
    \|\mf{Hz}_0(t)\|> \left(\frac{B_\mf{d}+\gamma_\mf{d}}{\tilde{k}_0}\right)^\frac{m+1}{m}.
    $$
    given arbitrary $\gamma_\mf{d}>0$ it follows that
    $$
    \begin{aligned}
    \dot{V} &= -\tilde{k}_0\mf{z}_0^\top\mf{H}\sgn{\mf{Hz}_0}{\frac{m}{m+1}} + (\mf{H}\mf{z}_0)^\top\mf{d}(t) \leq -\tilde{k}_0\homm{\mf{Hz}_0}{\frac{m}{m+1}+1}+B_\mf{d}\|\mf{H}\mf{z}_0\|\\
    &\leq -\tilde{k}_0\|{\mf{Hz}_0}\|^{\frac{m}{m+1}+1}+B_\mf{d}\|\mf{H}\mf{z}_0\| =-\|\mf{H}\mf{z}_0\|(\tilde{k}_0\|\mf{Hz}_0\|^\frac{m}{m+1} - B_\mf{d})<-\gamma_\mf{d} \|\mf{H}\mf{z}_0\|\leq-c\gamma_\mf{d} \sqrt{2V}
    \end{aligned}
    $$
    with $c = \sqrt{\lambda_{\min}(\mf{H}^2)/\lambda_{\max}(\mf{H})}$. Therefore, choosing $\tilde{k}_0>\tilde{k}_{0,\mf{d}}:=(B_\mf{d}+\gamma_\mf{d})/\left(B_{\mf{Hz}_0}^\infty\right)^{\frac{m}{m+1}}$ makes $\dot{V}<-c\gamma_\mf{d}\sqrt{2V}$ until the region $\|\mf{H}\mf{z}_0(t)\|\leq B_{\mf{H}\mf{z}_0}^\infty$ is reached after some finite time $\delta_\mf{d}>0$, leaving such region invariant. Moreover, due to $\dot{V}<-c\gamma_\mf{d}\sqrt{2V}$, $\delta_\mf{d}>0$ can be made arbitrarily small by making taking $\gamma_\mf{d}>0$ sufficiently big, and similarly for $\tilde{k}_{0,\mf{d}}$. The last part of the lemma follows by taking $c_0 = (B_\mf{d}+\gamma_\mf{d})^{(m+1)/m}$.
\end{proof}
\begin{lemma}
\label{le:derivative}
Consider the same assumptions as in Lemma \ref{le:state}. Additionally, assume $\|\dot{\mf{d}}(t)\|\leq B_{\dot{\mf{d}}}>0, \forall t\geq 0$. Given any $\delta_{\dot{\mf{d}}}, \overline{\theta}>0$ then, there exist $\tilde{k}_{0,{\dot{\mf{d}}}}\in\mathbb{R}$ such that $\|\dot{\mf{z}}_0(t)\|\leq \overline{\theta}$ for all $t\geq \delta_{\dot{\mf{d}}}+t_0$.
\end{lemma}
\begin{proof}
    Denote with $\mf{v} = \dot{\mf{z}}_0$ and with $\mf{J}(t) = \frac{m}{m+1}\text{diag}\left(|\mf{Hz}_0(t)|^{\frac{-1}{m+1}}\right)$. Therefore,
    $$
    \dot{\mf{v}} = \dot{\mf{d}}(t)-\frac{\tilde{k}_0m}{m+1}\left(|\mf{H}\mf{z}_0|^{-\frac{1}{m+1}}\right)\odot\mf{H}\mf{v} = \dot{\mf{d}}(t) - \tilde{k}_0\mf{J}(t)\mf{H}\mf{v}.
    $$
    Moreover, choose $\tilde{k}_0>\tilde{k}_{0,\mf{d}}$ as in Lemma \ref{le:state} such that $\|\mf{w}(t)\|\leq B_{\mf{Hz}_0}^\infty$ for all $t\geq \delta_\mf{d}$ with $\mf{w}(t)=\mf{H}\mf{z}_0(t)$. Hence, the components of $\mf{w}$ satisfy $|w_i|^\frac{1}{m+1}\leq \|\mf{w}\|^\frac{1}{m+1}\leq (B_{\mf{Hz}_0}^\infty)^\frac{1}{m+1}$ and therefore:
    $$
    \mf{J}(t)\succeq \frac{m\left(B_{\mf{Hz}_0}^\infty\right)^{-\frac{1}{m+1}}}{m+1}\mf{I}.
    $$
    
    Choose a Lyapunov function candidate $V(\mf{v}) = \mf{v}^\top\mf{H}\mf{v}/2$ such that:
    $$
    \begin{aligned}
        \dot{V} &= -\tilde{k}_0\mf{v}^\top\mf{H}\mf{J}(t)\mf{H}\mf{v} + \mf{v}^\top\mf{H}\dot{\mf{d}}(t)\leq -\frac{\tilde{k}_0m\left(B_{\mf{Hz}_0}^\infty\right)^{-\frac{1}{m+1}}}{m+1}\|\mf{H}\mf{v}\|^2 + B_{\dot{\mf{d}}}\|\mf{H}\mf{v}\|\\
        &\leq -\frac{\tilde{k}_0\lambda_{\min}(\mf{H})m\left(B_{\mf{Hz}_0}^\infty\right)^{-\frac{1}{m+1}}}{m+1}\|\mf{v}\|^2 + B_{\dot{\mf{d}}}\lambda_{\max}(\mf{H}^2)\|\mf{v}\| \leq -\|\mf{v}\|\left(\frac{\tilde{k}_0\lambda_{\min}(\mf{H})m\left(B_{\mf{Hz}_0}^\infty\right)^{-\frac{1}{m+1}}}{m+1}\|\mf{v}\|-B_{\dot{\mf{d}}}\lambda_{\max}(\mf{H}^2)\|\right)\\
        &< -\gamma_{\dot{\mf{d}}}\|\mf{v}\| \leq -\gamma_{\dot{\mf{d}}}\sqrt{2\lambda_{\max}(\mf{H})^{-1}V}
    \end{aligned}
    $$
    for arbitrary $\gamma_{\dot{\mf{d}}}>0$, and any $t\geq t_0$ such that:
    $$
    \|\mf{v}(t)\|=\|\dot{\mf{z}}_0(t)\| > \frac{(B_{\dot{\mf{d}}}\lambda_{\max}(\mf{H}^2)+\gamma_{\dot{\mf{d}}})(m+1)}{\tilde{k}_0\lambda_{\min}(\mf{H})m\left(B_{\mf{Hz}_0}^\infty\right)^{-\frac{1}{m+1}}}
    $$
    Henceforth, choose 
    $$
    \tilde{k}_{0,\dot{\mf{d}}}=\max\left(\frac{(B_{\dot{\mf{d}}}\lambda_{\max}(\mf{H}^2)+\gamma_{\dot{\mf{d}}})(m+1)\left(B_{\mf{Hz}_0}^\infty\right)^\frac{1}{m+1}}{\lambda_{\min}(\mf{H})m\overline{\theta}},\tilde{k}_{0,\mf{d}}\right)
    $$
    leading to $\dot{V}<-\gamma_{\dot{\mf{d}}}\sqrt{2\lambda_{\max}(\mf{H})^{-1}V}$ until the region $\|\dot{\mf{z}}_0(t)\|\leq \overline{\theta}$ is reached after some finite time $\delta_{\dot{\mf{d}}}>0$, leaving such region invariant. Similarly as before, $\delta_{\dot{\mf{d}}}>0$ can be made arbitrarily small as well by increasing $\tilde{k}_{0,\dot{\mf{d}}}>0$ further.
\end{proof}

\begin{lemma}
\label{le:contraction:leader}
    Let Assumption \ref{as:graph} hold and consider the trajectories of \eqref{eq:recursive}. Then, for any $B_{\mf{z}_\mu}^\infty,B_{\mf{z}_\mu}^{t_0}>0$ with $B_{\mf{z}_\mu}^\infty<B_{\mf{z}_\mu}^{t_0}$, there exists $\tilde{k}_0^*,\dots,\tilde{k}_m^*,{\delta}_\mf{z}>0$ such that if $\tilde{k}_\mu\geq \tilde{k}_\mu^*$ and the initial conditions satisfy $\|\mf{z}_\mu(t_0)\|\leq B_{\mf{z}_\mu}^{t_0}$, then $\|\mf{z}_\mu(t)\|\leq B_{\mf{z}_\mu}^\infty$ for all $t\geq {\delta}_\mf{z}$. 
\end{lemma}
\begin{proof}
    We follow a similar procedure as the one outlined in \cite[Lemma 13]{aldana2021}. For clarity in this proof, we divide this procedure in two parts. First, we provide some parameter definitions.

    Choose arbitrary $\delta_{\mfs{hom}}>0$. Set $\delta_{\infty}\in(0,\delta_{\mfs{hom}})$ and $\delta_{\mf{d}}=\delta_{\dot{\mf{d}}}\in(0,(\delta_{\mfs{hom}}-\delta_\infty)/2)$. Given $\delta_{\mfs{hom}}$ and  $B_{\mf{z}_\mu}^{t_0}$, set $B_{\mf{z}_0},\dots,B_{\mf{z}_m}$ as in Lemma \ref{le:zerobound}. Lemma \ref{le:zerobound} ensures bounds of the form $\|\mf{z}_\mu(\delta_{\mf{d}}+\delta_{\dot{\mf{d}}})\|\leq B_{\mf{z}_\mu}^{t_0=\delta_{\mf{d}}+\delta_{\dot{\mf{d}}}}>0$.
        These are used as input bounds for the initial conditions at $t_0=\delta_\mf{d}+\delta_{\dot{\mf{d}}}$ in Lemma \ref{prop:levant_contraction}, along with $B_{\mf{z}_\mu}^\infty$. Then, set $\tilde{k}_1,\dots,\tilde{k}_m, \overline{\theta}>0$ as the ones resulting from Lemma \ref{prop:levant_contraction}.  Setting $B_\mf{d}:=B_{\mf{z}_1}$, choose $c_0$ as in Lemma \ref{le:state}.  Moreover, choose $$B_{\dot{\mf{d}}} = \tilde{k}_1c_0^{\frac{m-1}{m+1}} + B_{\mf{z}_2}.$$ Set $B_{\mf{Hz}_0}^\infty=B_{\mf{z}_0}^\infty/\lambda_{\min}(\mf{H})^2$ such that $\|\mf{Hz}_0\|\leq B_{\mf{Hz}_0}^\infty$ imply $\|\mf{z}_0\|\leq B_{\mf{z}_0}^\infty$. Given $\delta_\mf{d}=\delta_{\dot{\mf{d}}}, \overline{\theta}, B_{\mf{Hz}_0}^\infty$ as defined before, set $\tilde{k}_{0,\dot{\mf{d}}}$ as in Lemma \ref{le:derivative}.

    Now, with these definitions, we proceed with the following steps. First, note that Lemma \ref{le:zerobound} ensures $\|\mf{z}_\mu(t)\|\leq B_{\mf{z}_\mu}$ for all $t\in[0,\delta_{\mfs{hom}}]$. Set $\mf{d}(t):=\mf{z}_1(t)$. This ensures $\|\mf{d}(t)\|\leq B_\mf{d}$  for all $t\in[\delta_\mf{d},\delta_{\mfs{hom}}]$. During $t\in[\delta_\mf{d},\delta_{\mfs{hom}}]$, Lemma \ref{le:state} ensures that $\|\mf{Hz}_0(t)\|\leq c_0\tilde{k}_0^{-\frac{m+1}{m}}$. This fact, in addition to $k_1=\tilde{k}_1 \tilde{k}_0^{\frac{m-1}{m}}$ allows to conclude:
        $$
        \begin{aligned}
        \|\dot{\mf{d}}(t)\|=\|\dot{\mf{z}}_1(t)\|&\leq {k}_1\|\sgn{\mf{H}\mf{z}_0}{\frac{m-1}{m+1}}\|+\|\mf{z}_{2}\| \leq \tilde{k}_1 \tilde{k}_0^{\frac{m-1}{m}} \|\mf{Hz}_0\|^{\frac{m-1}{m+1}} + B_{\mf{z}_2} \\
        &\leq \tilde{k}_1 \tilde{k}_0^{\frac{m-1}{m}} \left(c_0\tilde{k}_0^{-\frac{m+1}{m}} \right)^{\frac{m-1}{m+1}} + B_{\mf{z}_2} = \tilde{k}_1c_0^{\frac{m-1}{m+1}} + B_{\mf{z}_2} = B_{\dot{\mf{d}}}
    \end{aligned}
    $$
    for all $t\in[\delta_{\mf{d}},\delta_{\mfs{hom}}]$. As a result of the previous arguments points, and the definitions made so far including $\bar{\theta}$, Lemmas \ref{le:state} and \ref{le:derivative} ensures that $\|\mf{z}_0(t)\|\leq B_{\mf{z}_0}^\infty$ and $\|\dot{\mf{z}}_0(t)\|\leq \overline{\theta}$ for all $t\in[\delta_{\mf{d}}+\delta_{\dot{\mf{d}}},\delta_{\mfs{hom}}]$, due to the choice of $\tilde{k}_0$ and since $B_\mf{d},_{\dot{\mf{d}}}$ are independent of $\tilde{k}_0$.  In the interval $t\in[\delta_{\mf{d}}+\delta_{\dot{\mf{d}}},\delta_{\mfs{hom}}]$, $\dot{\mf{z}}_0(t)$ plays the role of $\theta(t)$ in \eqref{eq:levant_recursive}. Moreover, note that $\delta_{\mf{d}}+\delta_{\dot{\mf{d}}} + \delta_\infty < (\delta_{\mfs{hom}}-\delta_\infty) + \delta_\infty = \delta_{\mfs{hom}}$. Therefore, due to the initial conditions at $t_0=\delta_\mf{d}+\delta_{\dot{\mf{d}}}$, Lemma \ref{prop:levant_contraction} is used to conclude that $\|\mf{z}_\mu(t)\|\leq B_{\mf{z}_\mu}^\infty$ for all $t\in[ \delta_{\mf{d}}+\delta_{\dot{\mf{d}}} + \delta_\infty,  \delta_{\mfs{hom}})$ which is a non empty interval. These inequalities remain invariant for all $t\geq \delta_{\mfs{hom}}$.
\end{proof}

Now, we are ready to show Theorem \ref{th:main} for $\mu>1$. First, note that $\mf{r}$-homogeneity of \eqref{eq:main:z} imply that the system is invariant to the change of variables
\begin{equation}
\label{eq:transformation}
t'=\eta t, \quad \mf{z}_\mu'(t')=\eta^{m-\mu+1}\mf{z}_\mu(t'/\eta).
\end{equation}
Let $B_{\mf{z}_\mu}^{t_0}>0$ be a sufficiently big bound such that the initial conditions comply $\|\mf{z}_\mu(0)\|\leq B_{\mf{z}_\mu}^{t_0}$. Therefore, choose arbitrary $\eta\in(0,1)>0$ making $\eta_\mu=\eta^{m-\mu+1}<1$ and $B_{\mf{z}_\mu}^\infty=\eta_\mu B_{\mf{z}_\mu}^{t_0}<B_{\mf{z}_\mu}^{t_0}$. Choose $\{\tilde{k}_\mu\}_{\mu=0}^m,\delta_\mf{z}$ as in Lemma \ref{le:contraction:leader}, such that $\|\mf{z}_\mu(\delta_\mf{z})\|\leq B_{\mf{z}_\mu}^\infty$. In other words, there is a contraction from $\|\mf{z}_\mu(0)\|\leq B_{\mf{z}_\mu}^{t_0}$ to $\|\mf{z}_\mu(\delta_\mf{z})\|\leq \eta^{m-\mu+1}B_{\mf{z}_\mu}^{t_0}<B_{\mf{z}_\mu}^{t_0}$ after a interval of length $\delta_\mf{z}$. Invariance under the transformation \eqref{eq:transformation} implies that a a similar contraction occurs from $\|\mf{z}_\mu(\delta_\mf{z})\|\leq \eta^{m-\mu+1}B_{\mf{z}_\mu}^{t_0}$ to $\|\mf{z}_\mu(\delta_\mf{z}+\eta\delta_\mf{z})\|\leq \eta^{m-\mu+1}(\eta^{m-\mu+1}B_{\mf{z}_\mu}^{t_0})$ after a interval of length $\eta\delta$. This procedure is performed sequentially to conclude that after $s\geq 1$ steps:
$$
\|\mf{z}_\mu(\delta_\mf{z}(1+\eta+\eta^2+\dots+\eta^s))\|\leq \eta^{(s+1)(m-\mu+1)}B_{\mf{z}_\mu}^{t_0}.
$$
Finally, note that $\lim_{s\to\infty}\delta_\mf{z}(1+\eta+\eta^2+\dots+\eta^s)=\delta_\mf{z}/(1-\eta)=:T$ and
$
\lim_{s\to\infty}\eta^{(s+1)(m-\mu+1)}B_{\mf{z}_\mu}^{t_0}=0.
$
Therefore, $\|\mf{z}_\mu(t)\|=0$ for all $t\geq T$.

\section{Performance Analysis}
\subsection{Robustness analysis under bounded measurement noise}

Up to this point, it has been assumed that some followers have access to the leader signal $u(t)$ directly. However, in practice, it is usual for followers to have access to a noisy version of $u(t)$ instead, due to communication or measurement errors. In the following, we provide an extension of Theorem~\ref{th:main} to consider bounded measurement noise.

\begin{proposition}
\label{prop:noise}
Consider the system:
\begin{equation}
\label{eq:main:noise}
\begin{array}{rl}
\dot{x}_{i,\mu} &= {x}_{i,\mu+1}-k_\mu \tilde{L}^{\frac{\mu+1}{m+1}}\sgn{\displaystyle\sum_{j\in\mathcal{N}_i}({x}_{i,0}-{x}_{j,0})+b_i({x}_{i,0}-u_i)}{\frac{m-\mu}{m+1}}\quad\quad \text{for\ }  \mu=0,\dots,m-1 \\
\dot{x}_{i,m} &= -k_m\tilde{L}\sgn{\displaystyle\sum_{j\in\mathcal{N}_i}({x}_{i,0}-{x}_{j,0})+b_i({x}_{i,0}-u_i)}{0},
\end{array}
\end{equation}
where $u_i(t)=u(t)+\varepsilon_i(t)$ for some noise signals $\varepsilon_i(t)\in\overline{\varepsilon}[-1,1]$ and a known noise bound $\overline{\varepsilon}>0$. Let the gains $\{k_\mu\}_{\mu=0}^m$ be selected as in Theorem~\ref{th:main}. Then, there exist positive constants $\{c_\mu\}_{\mu=0}^m$ such that for every set of initial conditions, there exists $T>0$ such that for all $t\geq T$,
\begin{equation}
\label{eq:bounds}
    |x_{i,\mu}(t)-u^{(\mu)}(t)|\leq c_\mu(\overline{\varepsilon})^{\frac{m-\mu+1}{m+1}}.
\end{equation}
\end{proposition}
\begin{proof}
    Note that if $\overline{\varepsilon}=0$, then \eqref{eq:main:noise} is reduced to \eqref{eq:main} for which convergence in finite time to $x_{i,\mu}(t)=u^{(\mu)}(t)$ is ensured by Theorem~\ref{th:main}. Hence, homogeneity of \eqref{eq:main} allows us to apply \cite[Theorem 2]{levant2005} to obtain the desired asymptotical bounds.
\end{proof}

\begin{remark}
    Note that the result in Proposition \ref{prop:noise} is similar to what is obtained for a standard (non-distributed) RED in \cite{levant2003}. This is, exact differentiation is recovered as $\overline{\varepsilon}$ tends to zero (i.e., the distributed differentiator is robust in the sense of~\cite{levant1998robust}) and homogeneous accuracy bounds are obtained as in \eqref{eq:bounds} for $\overline{\varepsilon}>0$ which have the typical shape expected for a RED.
\end{remark}

\subsection{Sampled data implementation}
Inspired by discretization procedures in \cite{livne2014} for the RED, we provide a proper discretization of the proposal in \eqref{eq:main}, in the sense that it retains appropriate asymptotic behavior with chattering reduction for small sampling steps. Given a sampling step $\Delta>0$ consider the protocol:
\begin{equation}
\label{eq:discrete}
\begin{aligned}
&{x}_{i,\mu}[k+1] = \sum_{\nu=0}^{m-\mu}\left(\frac{\Delta^\nu}{\nu!}\right)x_{i,\mu+\nu}[k] -\Delta k_\mu \tilde{L}^{\frac{\mu+1}{m+1}}\sgn{\displaystyle\sum_{j\in\mathcal{N}_i}({x}_{i,0}[k]-{x}_{j,0}[k])+b_i({x}_{i,0}[k]-u[k])}{\frac{m-\mu}{m+1}}\hspace{1em}\text{for\ }  \mu=0,\dots,m-1 \\
&{x}_{i,m}[k+1] ={x}_{i,m}[k] -\Delta k_m\tilde{L}\sgn{\displaystyle\sum_{j\in\mathcal{N}_i}({x}_{i,0}[k]-{x}_{j,0}[k])+b_i({x}_{i,0}[k]-u[k])}{0},
\end{aligned}
\end{equation}
where $u[k]:=u(\Delta k)$. In the following we study convergence and accuracy of the sampled-data proposal in \eqref{eq:discrete}.

\begin{remark}
While this section does not consider measurement noise for the sake of simplicity, same as in \cite{livne2014}, similar accuracy bounds as in Proposition \ref{prop:noise} are obtained in the sampled-data case as well, as a consequence of homogeneity.
\end{remark}

\begin{proposition}
    Consider the assumptions and gain design of Theorem \ref{th:main}. Then, there exist positive constants $\{c_\mu\}_{\mu=0}^m$ such that for every set of initial conditions there exists $K\in\mathbb{N}$ such that the solutions of \eqref{eq:discrete} comply with
    \begin{equation}
    \label{eq:bounds:delta}
    |x_{i,\mu}[k] - u^{(\mu)}[k]|\leq c_\mu\Delta^{m-\mu+1}
    \end{equation}
    for all $k\geq K$ and $\mu = 1,\ldots, m$.
\end{proposition}

\begin{proof}
    First, write \eqref{eq:discrete} in vector form as:
\begin{equation}
\label{eq:discrete:vector}
\begin{aligned}
&\mf{x}_{\mu}[k+1] = \sum_{\nu=0}^{m-\mu}\left(\frac{\Delta^\nu}{\nu!}\right)\mf{x}_{\mu+\nu}[k] -\Delta k_\mu \tilde{L}^{\frac{\mu+1}{m+1}}\sgn{\mf{L}\mf{x}_0[k]+\mf{B}(\mf{x}_0[k]-\mathds{1}u[k])}{\frac{m-\mu}{m+1}} \text{\quad\quad for\ }  \mu=0,\dots,m-1 \\
&\mf{x}_{m}[k+1] =\mf{x}_{m}[k] -\Delta k_m\tilde{L}\sgn{\mf{L}\mf{x}_0+\mf{B}(\mf{x}_0-\mathds{1}u[k])}{0}.
\end{aligned}
\end{equation}
Similarly as before, define the error variable as $\mf{e}_\mu[k]= \mf{x}_\mu[k]-\mf{H}^{-1}\mf{B}\mathds{1}u^{(\mu)}[k]$. Note that 
\begin{equation}
\label{eq:input:expansion}
\begin{aligned}
    \mf{H}^{-1}\mf{B}\mathds{1}u^{(\mu)}[k+1]= \sum_{\nu=0}^{m-\mu}\left(\frac{\Delta^\nu}{\nu!}\right)\mf{H}^{-1}\mf{B}\mathds{1}u^{(\mu+\nu)}[k] + \mf{r}_{\mu}[k],
\end{aligned}
\end{equation}
with $\mf{r}_{\mu}[k]$ as in \eqref{eq:remainder} in Appendix \ref{sec:aux}. Subtract \eqref{eq:input:expansion} to \eqref{eq:discrete:vector} obtaining:
\begin{equation}
\label{eq:discrete:vector:error}
\begin{aligned}
&\mf{e}_{\mu}[k+1] = \sum_{\nu=0}^{m-\mu}\left(\frac{\Delta^\nu}{\nu!}\right)\mf{e}_{\mu+\nu}[k] -\Delta k_\mu \tilde{L}^{\frac{\mu+1}{m+1}}\sgn{\mf{H}\mf{e}_0[k]}{\frac{m-\mu}{m+1}}+\mf{r}_\mu[k]\quad\quad\text{for\ }  \mu=0,\dots,m-1 \\
&\mf{e}_{m}[k+1] \in\mf{e}_{m}[k] -\Delta k_m\tilde{L}\sgn{\mf{H}\mf{e}_0[k]}{0}+\Delta[-\tilde{L},\tilde{L}]^\mfs{N} \\
\end{aligned}
\end{equation}
Define
$$
\begin{aligned}
    &\mf{g}_\mu(\mf{e}[k];\Delta)=\sum_{\nu=1}^{m-\mu}\left(\frac{\Delta^{\nu-1}}{\nu!}\right)\mf{e}_{\mu+\nu}[k] - k_\mu \tilde{L}^{\frac{\mu+1}{m+1}}\sgn{\mf{H}\mf{e}_0[k]}{\frac{m-\mu}{m+1}}+\frac{\mf{r}_\mu[k]}{\Delta}\text{\quad\quad for }\mu=0,\dots,m-1\\
    &\mf{g}_m(\mf{e}[k];\Delta)=k_m\tilde{L}\sgn{\mf{H}\mf{e}_0[k]}{0}+[-\tilde{L},\tilde{L}]^\mfs{N},
\end{aligned}
$$
so that \eqref{eq:discrete:vector:error} can be written as:
\begin{equation}
\label{eq:discrete:vector:error:inclusion}
\begin{aligned}
&\mf{e}_{\mu}[k+1] \in \mf{e}_{\mu}[k]+\Delta\mf{g}_\mu(\mf{e}[k];\Delta)
\end{aligned}
\end{equation}
for all $\mu=0,\dots,m$. Let:
$$
\mf{e}_\mu(t)=\mf{e}_\mu[k] + (t-\Delta)\mf{v}_\mu[k]
$$
for some $\mf{v}_\mu[k]\in\mf{g}_\mu(\mf{e}[k];\Delta)$ and $t\in[k\Delta,(k+1)\Delta]$. Thus, $\mf{e}_\mu(t)$ satisfy \eqref{eq:discrete:vector:error:inclusion} at $t=(k+1)\Delta$ as well as:
\begin{equation}
\label{eq:delay:inclusion}
\dot{\mf{e}}_\mu(t)\in \mf{g}_\mu(\mf{e}(t-[0,\Delta]);\Delta)
\end{equation}
since $\mf{e}_\mu[k]=\mf{e}_\mu(t-(t-k\Delta))\in\mf{e}_\mu(t-[0,\Delta])$, for $\mu=0,\dots,m$. Note that the inclusion in \eqref{eq:delay:inclusion} is equivalent to \eqref{eq:main:e} when $\Delta=0$. Henceforth, Theorem \ref{th:main} ensures finite time convergence towards the origin when $\Delta=0$. Hence, \eqref{eq:delay:inclusion} is in the form required by \cite[Theorem 2]{levant2005}, where due to homogeneity it can be concluded that $\|\mf{e}_\mu(t)|\|\leq c_\mu'\Delta^{m-\mu+1}$ for $t\geq T$ for some $T, c_\mu'>0$ in the case with $\Delta>0$. Hence, due to Proposition \ref{prop:equilibrium}, the bounds in \eqref{eq:bounds:delta} follow.

\end{proof}

\section{Numerical example}
\label{Sec:Num}
\subsection{Second order system}
\label{sec:second}
Consider the case with $m=1$. The protocol \eqref{eq:main} then takes the form:
\begin{equation}
\label{eq:second:order}
\begin{array}{rl}
\dot{x}_{i,0} &= {x}_{i,1}-k_0 \sqrt{\tilde{L}}\sgn{\displaystyle\sum_{j\in\mathcal{N}_i}({x}_{i,0}-{x}_{j,0})+b_i({x}_{i,0}-u)}{\frac{1}{2}}\\
\dot{x}_{i,1} &= -k_1\tilde{L}\sgn{\displaystyle\sum_{j\in\mathcal{N}_i}({x}_{i,0}-{x}_{j,0})+b_i({x}_{i,0}-u)}{0}.
\end{array}
\end{equation}
The sampled-data implementation of \eqref{eq:main} according to \eqref{eq:discrete} corresponds to:
\begin{equation}
\begin{aligned}
&{x}_{i,0}[k+1] = {x}_{i,0}[k]+\Delta x_{i,1}[k] -\Delta k_0 \sqrt{\tilde{L}}\sgn{\displaystyle\sum_{j\in\mathcal{N}_i}({x}_{i,0}[k]-{x}_{j,0}[k])+b_i({x}_{i,0}[k]-u[k])}{\frac{1}{2}}\\
&{x}_{i,1}[k+1] ={x}_{i,1}[k] -\Delta k_1\tilde{L}\sgn{\displaystyle\sum_{j\in\mathcal{N}_i}({x}_{i,0}[k]-{x}_{j,0}[k])+b_i({x}_{i,0}[k]-u[k])}{0},
\end{aligned}
\end{equation}
which coincides with a forward Euler discretization of \eqref{eq:main}. For the simulations, we set $\Delta =10^{-3}s$. For the sake of simplicity, set the leader position as $u(t) = \sin(\omega t)$ with $\omega = 0.5$ such that $|u^{(m+1)}(t)|=|\ddot{u}(t)|\leq L$ with $L = \omega^2 = 0.25$. Consider a group of $\mfs{N}=10$ agents on a communication graph $\mathcal{G}$ with a single undirected cycle, i.e., $(j,i)\in\mathcal{E}$ iff $i-j=\pm1 \textrm{mod} \mfs{N}$, with agents in $\mathcal{L}=\{1,3,5\}$ having access to the leader's output $u(t)$ at discrete sampling-times $t=k\Delta$, $k=0,1,\ldots$. This sets $\tilde{L} = \mfs{N}L\lambda_{\max}(\mf{H}^{-1}\mf{B}) = 2.5$. In addition, one can set $k_0=2, k_1=1.1$ which can be numerically shown to comply with the conditions of Theorem \ref{th:main}. Initial conditions $x_{i,0}(0),x_{i,1}(0)$ were chosen arbitrarily in the interval $[-5,5]$. 

The trajectories of the states \(x_{i,\mu}(t)\), at times \(t=0, \Delta, 2\Delta, \ldots\), are depicted in the left column of Figure \ref{fig:second:order}. All agents achieve finite-time consensus on both the leader's signal \(u(t)\) and its derivative \(\dot{u}(t)\), despite not all agents having access to \(u(t)\). Additionally, noise was introduced to the agents with access to the leader's output as \(u_i(t) = u(t) + \varepsilon_i(t)\), where \(\varepsilon_i(t)\) was randomly drawn from an uniform distribution over \(\overline{\varepsilon}[-1,1]\) with \(\overline{\varepsilon} = 0.1\), for each \(t = k\Delta\). The right column of Figure \ref{fig:second:order} illustrates the noisy measurements \(u_i(t)\) and compares them with the estimations from \(x_{i,0}(t)\), showcasing the filtering properties of the observer. Moreover, \(x_{i,1}(t)\) closely approximates the derivative \(\dot{u}(t)\) despite the presence of noise.

\begin{figure}
\centering
\includegraphics[width=0.8\textwidth]{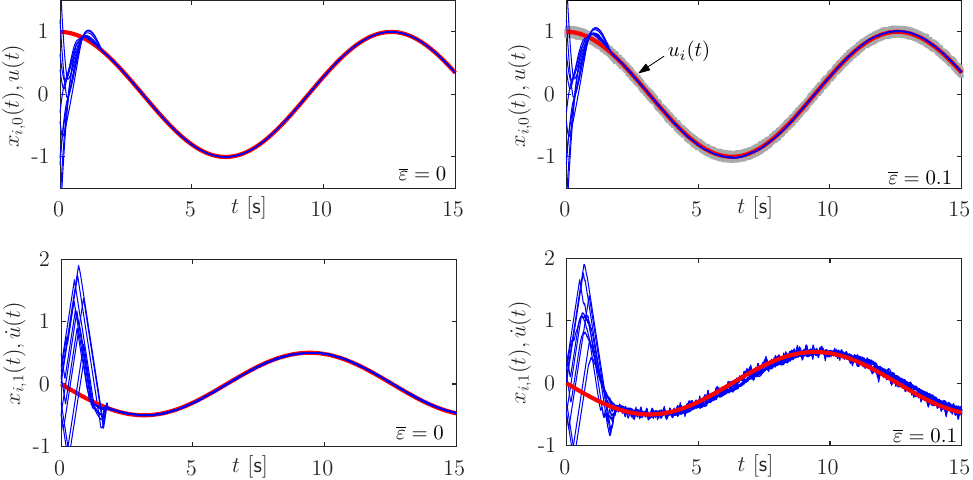}
\caption{Trajectories for the leader position $u(t)$ and its derivative $\dot{u}(t)$ (red), as well as the trajectories for $x_{i,0}(t), x_{i,1}(t)$ for all agents $i\in\mathcal{I}$ (blue) with $m=1$. Two scenarios are considered, namely $\overline{\varepsilon}=0$ (noiseless) and $\overline{\varepsilon} = 0.1$. Due to the sampled-data implementation, only $t=0, \Delta, 2\Delta, \dots$ with $\Delta=10^{-3}s$ are depicted.}
\label{fig:second:order}
\end{figure}

\subsection{High order system}
 Consider the same scenario as in the second order example. However, in this case, we set $m=3$ for a fourth-order leader system. Consequently, $|u^{(m+1)}(t)|=|u^{(4)}(t)| \leq L$ with $L = \omega^4=0.0625, \tilde{L} = 0.625$. The sampled-data implementation of \eqref{eq:main} using \eqref{eq:discrete} results in:

 \begin{equation}
\begin{aligned}
&{x}_{i,0}[k+1] = x_{i,0}[k] + \Delta x_{i,1}[k] + \frac{\Delta^2}{2}x_{i,2}[k] + \frac{\Delta^3}{6}x_{i,3}[k] &-\Delta k_0\tilde{L}^{\frac{1}{4}}\sgn{\displaystyle\sum_{j\in\mathcal{N}_i}({x}_{i,0}[k]-{x}_{j,0}[k])+b_i({x}_{i,0}[k]-u[k])}{\frac{3}{4}}\\
&{x}_{i,1}[k+1] = x_{i,1}[k]+ \Delta x_{i,2}[k] + \frac{\Delta^2}{2}x_{i,3}[k] &-\Delta k_1 \tilde{L}^{\frac{2}{4}}\sgn{\displaystyle\sum_{j\in\mathcal{N}_i}({x}_{i,0}[k]-{x}_{j,0}[k])+b_i({x}_{i,0}[k]-u[k])}{\frac{2}{4}} \\
&{x}_{i,2}[k+1] = x_{i,2}[k] + \Delta x_{i,3}[k] &-\Delta k_2 \tilde{L}^{\frac{3}{4}}\sgn{\displaystyle\sum_{j\in\mathcal{N}_i}({x}_{i,0}[k]-{x}_{j,0}[k])+b_i({x}_{i,0}[k]-u[k])}{\frac{1}{4}}\\
&{x}_{i,3}[k+1] ={x}_{i,3}[k] &-\Delta k_3\tilde{L}\sgn{\displaystyle\sum_{j\in\mathcal{N}_i}({x}_{i,0}[k]-{x}_{j,0}[k])+b_i({x}_{i,0}[k]-u[k])}{0} \\
\end{aligned}
\end{equation}
which is no longer a forward Euler discretization of \eqref{eq:main}. Following Theorem \ref{th:main}, we pick $\tilde{k}_1,\tilde{k}_2,\tilde{k}_3$ for a standard RED. Examples for these parameters can be found in \cite{livne2014} where we borrow $\tilde{k}_1=1.1, \tilde{k}_2=1.5, \tilde{k}_3=2$. Similarly as argued for the RED in \cite{levant2003,livne2014}, $\tilde{k}_0=50$ was found by simulation. Therefore, $k_0 = \tilde{k}_0=50, k_1 = \tilde{k}_1 k_0^{\frac{2}{3}}=14.92, k_2 = \tilde{k}_2 k_0^{\frac{1}{2}}=10.6, k_3 = \tilde{k}_3=2$.

The resulting trajectories for this experiment are shown in Figure \ref{fig:fourth:order}. Similarly, as before, we consider the cases with and without noise with $\overline{\varepsilon}=0$ and $\overline{\varepsilon}=0.1$ respectively. In the noise-free case, it can be observed that $x_{i,0}(t), x_{i,1}(t), x_{i,2}(t), x_{i,3}(t)$ can effectively estimate $u(t), \dot{u}(t), \ddot{u}(t), u^{(3)}(t)$ respectively after some finite-time, where some chattering is evident in $x_{i,3}(t)$ due to the sampled-data implementation. Similarly as in the example for the second order system, the estimates $x_{i,0}(t), x_{i,1}(t), x_{i,2}(t), x_{i,3}(t)$ still converge to a neighborhood around $u(t), \dot{u}(t), \ddot{u}(t), u^{(3)}(t)$ respectively, even in the presence of measurement noise.

\begin{figure}
\centering
\includegraphics[width=0.8\textwidth]{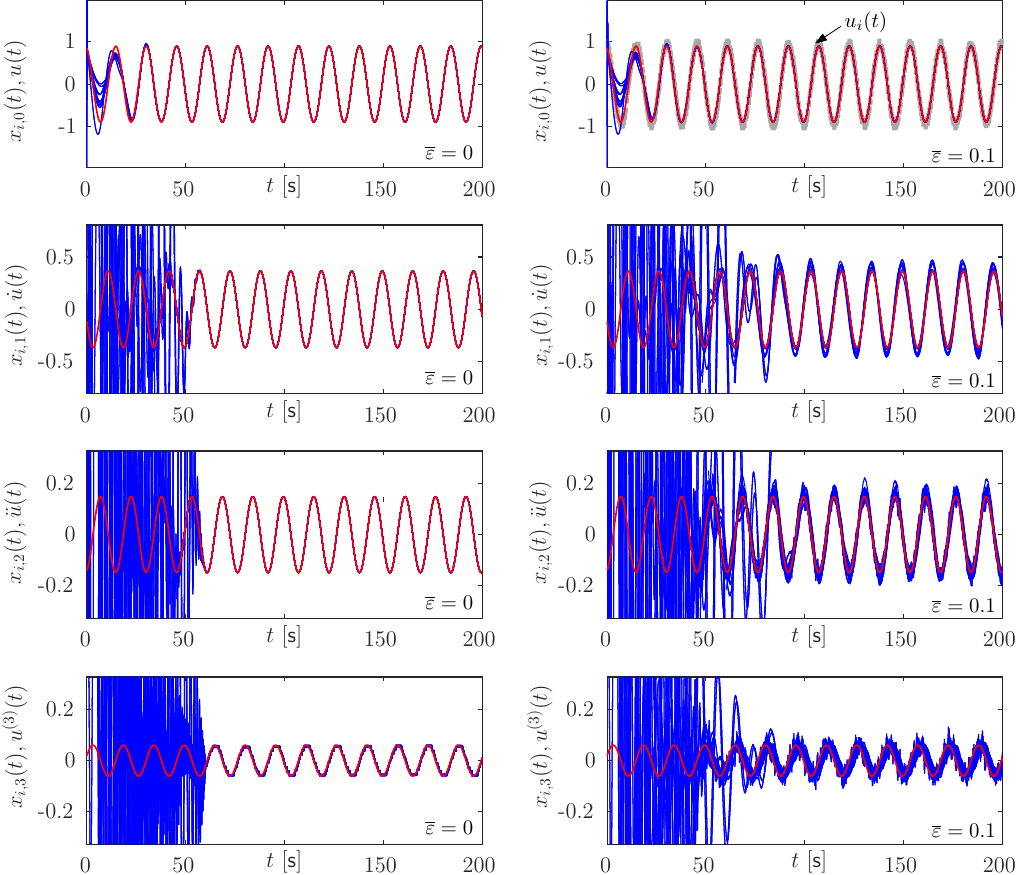}
\caption{Trajectories for the leader position $u(t)$ and its derivatives $\dot{u}(t),\ddot{u}(t),{u}^{(3)}(t)$ (red), as well as the trajectories for $x_{i,0}(t), x_{i,1}(t), x_{i,2}(t), x_{i,3}(t)$ for all agents $i\in\mathcal{I}$ (blue) with $m=3$. Two scenarios are considered, namely $\overline{\varepsilon}=0$ (noiseless) and $\overline{\varepsilon} = 0.1$. Due to the sample-data implementation, only $t=0, \Delta, 2\Delta, \dots$ with $\Delta=10^{-3}$ are depicted.}
\label{fig:fourth:order}
\end{figure}

\section{Conclusions}
\label{Sec:Concl}
We have introduced a novel distributed leader-observer that distinguishes itself by communicating only the observer output, rather than the full state, to neighboring agents, and by the fact that no agent needs to know whether it has communication with the leader or not. Due to the homogeneous structure of the algorithm, features similar to those for the standard Robust Exact Differentiator (RED) are exhibited by our proposal. In particular, in the ideal case of continuous-time communication and absence of measurement errors, exact estimation of the leader and its derivatives were obtained, thus performing what we call exact distributed differentiation. The proofs were based on a nontrivial adaptation and extension of techniques used for the RED. The framework admits to incorporating, for the first time in this area, measurement noise for the leader signal and a sampled-data implementation. The effectiveness of our approach has been numerically verified in both second-order and fourth-order systems, demonstrating its applicability and robustness across the estimation of derivatives of different orders. \Arxiv{As a future work, we consider the extension of the methods presented in this work to dynamic and directed networks. Moreover, taking into account the well-known properties of the variable structure control, the proposed approach can also be implemented to estimate an unknown input of the leader by resorting to the concept of the equivalent control, e.g., by filtering the discontinuous input of the observers, or by increasing the dimension of the observer state in case of Lipschitz unknown inputs.}

\section*{Acknowledgements}

H. Haimovich gratefully acknowledges the financial support provided by Agencia I+D+i PICT 2021-I-A-0730, Argentina. E. Usai gratefully acknowledges the financial backing provided by the research project “Network 4 Energy Sustainable Transition – NEST” funded under the National Recovery and Resilience Plan (NRRP) of Ministero dell’Università e della Ricerca (MUR); funded by the European Union – Next Generation EU, code PE0000021, CUP F53C2200077000.

\appendix
\section{Auxiliary results}
\label{sec:aux}

\begin{lemma}\cite{cao2012}
\label{le:H}
    Let Assumption \ref{as:graph} hold for a graph $\mathcal{G}$. Moreover, let $\mf{L}$ be the Laplacian matrix of $\mathcal{G}$ and $\mf{B}=\text{diag}(b_1,\dots,b_\mfs{N})$ with $b_i\in\{0,1\}$. Then, $\mf{L}+\mf{B}$ is positive definite. 
\end{lemma}

\begin{lemma}[Schur complement lemma]\cite[Theorem 7.7.7]{horn} 
\label{le:schur}
\Rev{Let 
$$
\mf{M} = \begin{bmatrix}
    \mf{A} & \mf{B} \\
    \mf{B}^\top & \mf{C}
\end{bmatrix}
$$
be symmetric with symmetric $\mf{A}\in\mathbb{R}^{p\times p}$, $\mf{C}\in\mathbb{R}^{q\times q}$. Then, $\mf{M}$ is positive definite if and only if $\mf{A}$ and its Schur complement $\mf{C}-\mf{B}^\top\mf{A}^{-1}\mf{B}$ are positive definite.}
\end{lemma}
\begin{lemma}\cite[Corollary 19]{aldana2021}
\label{le:sgn}
Given any $\mf{v}\in\mathbb{R}^\mfs{N}$ and $0<\alpha<1$ then, it follows that 
\begin{enumerate}
    \item $\|\sgn{\mf{v}}{\alpha}\|\leq \mfs{N}^\frac{1-\alpha}{2}\|\mf{v}\|^\alpha$.
    \item $\homm{\mf{v}}{\alpha}\geq \|\mf{v}\|^\alpha$.
\end{enumerate}
\end{lemma}
\begin{lemma}
\label{le:ineq}
Given any $\alpha>0$ then, for every $\mf{v},\mf{w}\in\mathbb{R}^\mfs{N}$ it follows that $$\left(\mf{v}-\mf{w}\right)^\top\left(\sgn{\mf{v}}{\alpha}-\sgn{\mf{w}}{\alpha}\right)\geq 0$$ with equality if and only if $\mf{v}=\mf{w}$.
\end{lemma}
\begin{proof}
Let $\mf{v}=[v_1,\dots,v_\mfs{N}]^\top, \mf{v}=[w_1,\dots,w_\mfs{N}]^\top$. We prove the result by induction. As induction base $\mfs{N}=1$, assume the opposite, that:
$$
(v_1-w_1)(\sgn{v_1}{\alpha}-\sgn{w_1}{\alpha})< 0.
$$
Consider the case $v_1-w_1>0$ equivalently $v_1>w_1$. Hence, $\sgn{v_1}{\alpha}< \sgn{w_1}{\alpha}$ which is a contradiction since $\sgn{\bullet}{\alpha}$ is strictly increasing. The case $v_1-w_1<0$ follows in the same way, establishing that $(v_1-w_1)(\sgn{v_1}{\alpha}-\sgn{w_1}{\alpha})\geq 0$. Moreover, assume that $(v_1-w_1)(\sgn{v_1}{\alpha}-\sgn{w_1}{\alpha})=0$ for $v_1\neq w_1$. This implies $\sgn{v_1}{\alpha}=\sgn{w_1}{\alpha}$ which is a contradiction.

Now, assume that the result is true for vectors of dimension $\mfs{N}-1$. Hence,
$$
\begin{aligned}
&\left(\mf{v}-\mf{w}\right)^\top\left(\sgn{\mf{v}}{\alpha}-\sgn{\mf{w}}{\alpha}\right) = \sum_{i=1}^\mfs{N} (v_i-w_i)(\sgn{v_i}{\alpha}-\sgn{w_i}{\alpha})=(v_\mfs{N}-w_\mfs{N})(\sgn{v_\mfs{N}}{\alpha}-\sgn{w_\mfs{N}}{\alpha})+ \left(\mf{v}'-\mf{w}'\right)^\top\left(\sgn{\mf{v}'}{\alpha}-\sgn{\mf{w}'}{\alpha}\right) \geq 0
\end{aligned}
$$
where $\mf{v}' = [v_1,\dots,v_{\mfs{N}-1}]^\top, \mf{w}' = [w_1,\dots,w_{\mfs{N}-1}]^\top$ and combining the base case and the assumption for dimension $\mfs{N}-1$. Now, assume that $\left(\mf{v}-\mf{w}\right)^\top\left(\sgn{\mf{v}}{\alpha}-\sgn{\mf{w}}{\alpha}\right)=0$. Hence,
$$
(v_\mfs{N}-w_\mfs{N})(\sgn{v_\mfs{N}}{\alpha}-\sgn{w_\mfs{N}}{\alpha})=-\left(\mf{v}'-\mf{w}'\right)^\top\left(\sgn{\mf{v}'}{\alpha}-\sgn{\mf{w}'}{\alpha}\right).
$$
However, both sides of the previous equation are non negative and thus it must be the case that $v_\mfs{N}=w_\mfs{N}$ and $\mf{v}'=\mf{w}'$ implying that $\mf{v}=\mf{w}$ necessarily.
\end{proof}
\begin{lemma}\cite[Lemma 4]{cruz2019}
\label{le:gain}
    Given $\eta:\mathbb{R}^n\to\mathbb{R}$ and $\gamma:\mathbb{R}^n\to\mathbb{R}$ two continuous homogeneous functions, with weights $\mf{r}=[r_1,\dots,r_n]$ and degrees $m$ with $\eta(\mf{x})\geq 0$, such that the following holds
    $$
    \{\mf{x}\in\mathbb{R}^n\setminus\{0\}:\eta(\mf{x})={0}\}\subseteq \{\mf{x}\in\mathbb{R}^n\setminus\{0\}:\gamma(\mf{x})<0\}.
    $$
Then, there exist $k^*>0$ such that for any $k\geq k^*$ it follows that
$$
-k\eta(\mf{x})+\gamma(\mf{x})<0.
$$
In particular, the previous inequality is satisfied if:
$$
k>k^*=\sup_{\|\mf{x}\|=1} \left(\frac{\gamma(\mf{x})}{\eta(\mf{x})}\right).
$$
\end{lemma}

\begin{lemma}\cite[Corollary A2]{aldana2023}
\label{lemma:taylor}
    Consider some integer $m$ such that $\mf{f}(t)\in\mathbb{R}^\mfs{N}$ is $m+1$ times differentiable for all $t\geq 0$. Moreover, assume that $\mf{f}^{(m+1)}(t)\in[-\tilde{L},\tilde{L}]^\mfs{N}, \forall t\geq 0$ with some $\tilde{L}$. Then, for any $\Delta>0$ it follows that
    \begin{equation}
    \label{eq:taylor}
{\mf{f}}^{(\mu)}[k+1] = \sum_{\nu=0}^{m-\mu}\left(\frac{\Delta^\nu}{\nu!}\right){\mf{f}}^{(\mu+\nu)}[k] + \mf{r}_{\mu}[k],
    \end{equation}
    where $\mf{f}[k]=\mf{f}(\Delta k)$, where the remainder satisfies
    \begin{equation}
    \label{eq:remainder}
    \mf{r}_{\mu}[k]\in\frac{\tilde{L}\Delta^{m-\mu+1}}{(m-\mu+1)!}[-1,1]^\mfs{N}.
    \end{equation}
\end{lemma}

\section{Weighted homogeneity}
\label{sec:homo}
In this section, we consider set-valued fields $F:\mathbb{R}^n\rightarrow P(\mathbb{R}^n)$, where $P(\mathbb{R}^n)$ is the power set of $\mathbb{R}^n$. Let $\Delta_{\mathbf{r}}(\lambda) = \text{diag}([\lambda^{r_1},\dots,\lambda^{r_n}])$ where $\mathbf{r}=[r_1,\dots,r_n]$ are called the weights and $\lambda>0$. For any $\mf{x}\in\mathbb{R}^n$, the vector
$
\Delta_{\mathbf{r}}(\lambda)\mf{x} = [\lambda^{r_1}x_1,\dots,\lambda^{r_n}x_n]^T
$
is called its standard dilation (weighted by $\mathbf{r}$). The following are some definitions of interest regarding the so-called $\mf{r}$-homogenety with respect to the standard dilation.

\begin{definition}[Homogeneous scalar functions]\cite[Definition 4.7]{bernuau2014}
A scalar function $V:\mathbb{R}^n\to\mathbb{R}$ is said to be $\mathbf{r}$-homogeneous of degree $d$  if
$
V(\Delta_{\mathbf{r}}(\lambda)\mf{x}) = \lambda^dV(\mf{x})
$ for any $\mf{x}\in\mathbb{R}^n$.
\end{definition}

\begin{definition}[Homogeneous set-valued fields]\cite[Definition 4.20]{bernuau2014}\label{def:homo}
A set-valued vector field $F:\mathbb{R}^n\rightarrow P(\mathbb{R}^n)$ is said to be $\mathbf{r}$-homogeneous of degree $d$ if
$
F(\Delta_{\mathbf{r}}(\lambda)\mf{x}) = \lambda^d\Delta_{\mathbf{r}}(\lambda)F(\mf{x})
$ for any $\mf{x}\in\mathbb{R}^n$.
\end{definition}

\section{Recursive form of homogeneous differentiators}
\label{sec:homodif}
In this section, we recall some results used in \cite{levant2003} to show the stability of the Levant's arbitrary order robust exact differentiator. In particular, we are interested in the properties of the recursive system
\begin{equation}
\label{eq:levant_recursive}
    \begin{aligned}
    \dot{z}_1 &= z_{2} - \tilde{k}_1\sgn{z_1+\theta(t)}{\frac{m-1}{m}}\\[0.5em]
    \dot{z}_\mu &= z_{\mu+1} - \tilde{k}_\mu\sgn{z_\mu-\dot{z}_{\mu-1}}{\frac{m-\mu}{m-(\mu-1)}}\quad\quad\quad \text{\normalfont for } \mu=2,\dots,m-1 \\[0.5em]
    \dot{z}_m&\in-\tilde{k}_m\sgn{z_m(t)-\dot{z}_{m-1}(t)}{0}+[-1,1]
    \end{aligned}
\end{equation}
with $z_1(t),\dots,z_m(t)\in\mathbb{R}$ and the measurable map $\theta(t)\in\mathbb{R}$. An important result regarding the contraction property of \eqref{eq:levant_recursive} is given in the following.

\begin{lemma}[Contraction property of \eqref{eq:levant_recursive}]\cite[Lemma 8]{levant2003}\label{prop:levant_contraction}
Let $\theta:\mathbb{R}\to[-\bar{\theta},\bar{\theta}]$ with $\bar{\theta}>0$. Given any $B_{z_\mu}^{t_0},B_{z_\mu}^{\infty}>0$, with $B_{z_\mu}^{\infty}<B_{z_\mu}^{t_0}, \mu=1,\dots,m, \delta_\infty>0$ then, there exists $B_{z_\mu}>B_{z_\mu}^{t_0}$, some gains $\tilde{k}_1,\dots,\tilde{k}_m>0$ (sufficiently big) and $\bar{\theta}>0$ (sufficiently small) such that any trajectory of \eqref{eq:levant_recursive} satisfying $|z_\mu(t_0)|\leq B_{z_\mu}^{t_0}$ for arbitrary $t_0\geq 0$ will satisfy $|z_\mu(t)|\leq B_{z_\mu}, \forall \Elio{t\in[t_0,t_0+\delta_\infty)}$ and $|z_\mu(t)|\leq B_{z_\mu}^\infty, \forall t\geq t_0+\delta_\infty$. 
\end{lemma}

\end{document}